\documentclass[12pt]{article}

\title{Wake Up and Join Me! An Energy-Efficient Algorithm for Maximal Matching in Radio Networks}
    
\newcommand{\email}[1]{\emph{#1}}

\author{Varsha Dani\thanks{Department of Computer Science, Rochester Institute of Technology. \email{vdani@cs.rit.edu}} 
\and Aayush Gupta\thanks{Department of Computer Science, University of New Mexico. \email{\{aayush,hayes\}@unm.edu}  Partially supported by NSF grant CCF-1150281.} 
\and Thomas P. Hayes$^\dagger$ 
\and Seth Pettie\thanks{Department of Electrical Engineering \& Computer Science, University of Michigan. \email{pettie@umich.edu}  Partially supported by NSF grant CCF-1815316.}}

\date{}

\usepackage{fullpage}
\usepackage{algorithm}
\usepackage{amssymb}
\usepackage{amsmath, amsthm, amssymb}
\usepackage{algpseudocode}
\usepackage{graphicx}
\usepackage[dvipsnames]{xcolor}

\newtheorem{theorem}{Theorem}[section]
\newtheorem{definition}[theorem]{Definition}
\newtheorem{proposition}[theorem]{Proposition}
\newtheorem{lemma}[theorem]{Lemma}

\newcommand{\tmax}{T}
\newcommand{\Prob}[1]{\mathbb{P}\left(#1\right)}
\newcommand{\Exp}{\mathbb{E}}
\newcommand{\ExpCond}[2]{\Exp\left(\left. #1 \;\right| #2 \right)}
\newcommand{\indicator}[1]{ \mathbf{1} \!\left( #1 \right) }
\newcommand{\partner}{{\tt partner}}
\DeclareMathOperator{\polylog}{polylog}

\newcommand{\eps}{\varepsilon}

\begin{document}

\maketitle

\thispagestyle{empty}

\begin{abstract}
We consider networks of small, autonomous devices that communicate with each other wirelessly. Minimizing energy usage is an important consideration in designing algorithms for such networks, as battery life is a crucial and limited resource.  Working in a model where both sending and listening for messages deplete energy, we consider the problem of finding a maximal matching of the nodes in a radio network of arbitrary and unknown topology.  

We present a distributed randomized algorithm that produces, with high probability, a maximal matching.  
The maximum energy cost per node is $O\big((\log n)(\log \Delta)\big),$ and the time complexity is $O(\Delta \log n)$.
Here $n$ is any upper bound on the number of nodes, and $\Delta$ is any upper bound on the maximum degree;
$n$ and $\Delta$ are parameters of our algorithm that we assume are known \emph{a priori} to all the processors.
We note that there exist families of graphs for which our bounds on energy cost and time complexity are simultaneously
optimal up to polylog factors, so any significant improvement would need additional assumptions about the network topology.

We also consider the related problem of assigning, for each node in the network, a neighbor to back up its data 
in case of eventual node failure.  Here, a key goal is to minimize the maximum \emph{load}, defined as the number of nodes assigned
to a single node.  We present an efficient decentralized low-energy algorithm that 
finds a neighbor assignment whose maximum load is at most a polylog($n$) factor bigger that the optimum. 
\end{abstract}

\setcounter{page}{0}
\newpage
\section{Introduction}
\label{sec:intro}

For networks of small computers, energy management and conservation is often a major concern.
When these networks communicate wirelessly, usage of the radio transceiver to
send or listen for messages is often one of the dominant causes of energy usage.
Moreover, this has tended to be increasingly true as the devices have gotten smaller;
see, for example, \cite{polastre2005telos, barnes2010ens,heinzelman2000energy}.  
Motivated by these considerations, Chang \emph{et al.}~\cite{chang2018energy}
introduced a theoretical model of distributed computation in which each send or listen operation
costs one unit of energy, but local computation is free.  Over a sequence of discrete 
timesteps, nodes choose whether to sleep, listen, or send a message of $O(\log n)$ bits.
A listening node successfully receives a message only when exactly one of its neighbors
has chosen to send in that timestep; otherwise it receives no input.

It is not uncommon for research on sensor networks to make assumptions about the
topology of the network, such as assuming the network is defined by a unit disk graph, or
that each node is aware of its location using GPS.  However, we will be interested in the
more general setting where we make almost no assumptions about the network topology.
We will assume that communication takes place via radio broadcasts, and that there is 
an arbitrary and unknown undirected graph $G$ whose edges indicate which pairs of nodes are 
capable of hearing each other's broadcasts.  We will, however, assume that each node is
initialized with shared parameters $n$ and $\Delta$, which are upper bounds on, respectively,
the total number of nodes, and the maximum degree of any node.  By designing algorithms to
operate without pre-conditions on, or foreknowledge of, the network topology, we 
potentially broaden the possible applications of our algorithms, and, by extension, of sensor networks.
For instance, we can imagine a network of small sensors scattered rather haphazardly from
an airplane passing over hazardous terrain; the sensors that survive their landing are unlikely
to be placed predictably or uniformly.

In this model, \cite{chang2018energy} presented a polylog-energy, polynomial-time
algorithm for the problem of one-to-all broadcast.  A later paper by Chang, Dani, Hayes, and Pettie~\cite{chang2020energy}
gave a sub-polynomial ($n^{o(1)}$) energy, polynomial-time algorithm for the related problem of breadth-first search.  An earlier body of work examined energy complexity
in \emph{single hop} networks~\cite{BenderKPY18,ChangDJ21,ChangKPWZ19,JurdzinskiKZ02,JurdzinskiKZ02b,JurdzinskiKZ02c,kardas2013energy,NakanoO00}, 
i.e., in which the network topology is known to be a clique.

In the present work, we will be concerned with another fundamental problem of graph theory, namely to find large 
sets of pairwise disjoint edges, or \emph{matchings}.  The problem of finding large matchings
has been thoroughly studied in a wide variety of computational models dating back more than a
century, to K\"onig~\cite{konig1916graphen}.  
For a fairly comprehensive review of past results, we recommend Duan and Pettie~\cite[Section 1]{duan2014linear}.

The main goal of the present work is to present a polylog-energy, polynomial-time distributed
algorithm that computes a maximal matching in the network graph.  The term maximal here indicates 
that the matching
intersects every edge of the graph, and therefore cannot be augmented without first removing edges.
It is well-known that a maximal matching necessarily has at least half as many edges as the
largest, or ``maximum'' matching.  In fact, as we discuss in Section~\ref{sec:max-max}, maximal matchings
are often significantly closer to being maximum than the aforementioned fact would indicate.

\begin{theorem} \label{thm:main}
    Let $G$ be any graph on at most $n$ vertices, of maximum degree at most $\Delta$.  
    Then Algorithm~\ref{alg:main} always terminates in $O(\Delta \log n)$ timesteps, 
    at which point each node knows its partner in a matching, $M$.
    Furthermore,
    \[
    \Prob{M \mbox{ is a maximal matching and every node used energy } \le 2C (\log n)(\log \Delta)} \ge 1 - \frac{1}{n^2}.
    \]
\end{theorem}

Observe that the per-node energy use is polylog($n$), which obviously cannot be 
improved by more than 
a polylog factor.  Moreover, the time complexity bound, $O(\Delta \log n)$, is also nearly optimal, 
when one considers that $G$ could contain a clique of size $\Delta$, in which case, in order for 
all the nodes in that clique to get even one chance to send a message and have it received by the other
nodes in the clique, there must be at least $\Delta$ timesteps, since our model does not allow a node
the possibility to receive two or more messages in a single round.  To put this another way, when
$\Delta$ is small, a high degree of parallelism is possible, which our algorithm exploits; but, when
$\Delta$ is large, there exist graphs for which this parallelism is impossible.

\subsection{Application: Neighbor Assignment}
One possible motivation for finding large matchings, apart from their intrinsic mathematical interest,
comes from the desire to {\bf back up data in case of node failures}.  Suppose we had a perfect matching (that is,
one whose edges contain every node) on the $n$ nodes of our network.  Then the matching could be viewed as
pairing each node with a neighboring node that could serve as its backup device.  This would ensure
that each device has a load of one node to back up, and that each node is directly adjacent to its backup device.

Since perfect matchings are not always available, we consider a more general scheme, in which each node is
assigned one of its neighbors to be its backup device, but we allow for loads greater than one.
Such a function can be visualized as a directed graph, with a directed edge from each node to its backup device.
In this case, each node has out-degree $1$, and load equal to its in-degree.
We would like to minimize the maximum load over all vertices. 

In Section~\ref{sec:buddies}, we will show that, if one is willing to accept a maximum load that is $O(\log n)$
times the optimum, this problem can be simply reduced to the maximal matching problem.  In light of our main
result, this means that, if there exists a neighbor assignment with $\polylog(n)$ maximum load, then we can 
find one on a radio network, while using only $\polylog(n)$ energy.

\subsection{Techniques}
Our matching algorithm can be thought of as a distributed and low-energy version of
the following greedy, centralized algorithm.  Randomly shuffle the $m$ edges.  Then,
processing the edges in order, accept each edge that is disjoint from all previous edges.
Note that this always results in a maximal matching.

To make this into a distributed algorithm, we make each node, in parallel, try to establish 
contact with one of its unmatched neighbors to form an edge.  Since a node can only receive a
message successfully if exactly one of its neighbors is sending, we limit the probability for
each node to participate in a given round, by setting a participation rate that is, with high 
probability, at most the inverse of the maximum degree of the residual graph induced by the 
unmatched nodes.  It turns out that this can be accomplished using a set schedule, where the
participation rate is a function of the amount of elapsed time.

The main technical obstacle in the analysis is proving that the maximum degree of the graph
decreases according to schedule (or faster).  This is achieved by noting that, if not, the
first vertex to have its degree exceed the schedule would have to have been failed to be
paired by our algorithm, despite going through a long sequence of consecutive rounds in which 
its chance to be paired was relatively high.

\subsection{Related Work}

Multi-hop radio network models have a long history, going back at least to work in the early 1990's by Bar-Yehuda, Goldreich, Itai~\cite{bar1991efficient,bar1992time}
among others.  The particular model of energy-aware radio computation we are using was introduced by Chang et al.~\cite{chang2018energy}.

A recent result by Chatterjee, Gmyr, and Pandurangan~\cite{chatterjee2020sleeping} considered the closely related problem
of Maximal Independent Set in another model, called the ``Sleeping model.''  Although it has some interesting similarities to our
work, there are several important differences.  Firstly, we note that although matchings of $G$ are nothing more than independent sets
on the line graph of $G$, in distributed computing, we cannot just convert an algorithm designed to run on the line graph of $G$ into
an algorithm to run on $G$.  Secondly, we note that the Sleeping model is based on the CONGEST model, and so, when a node is awake, 
it is allowed to send a different message to each of its neighbors at a unit cost.  By contrast, in our model, one node can only send one
message in a timestep, and it may collide with messages sent by other nodes.

Moscibroda and Wattenhofer~\cite{moscibroda2005maximal} considered the problem of finding a Maximal Independent Set in a radio network.
Their work also has some interesting similarities to ours, although they are assuming a unit-disk topology, and listening for messages is free 
in their model.  On the other hand, their algorithm works even when the nodes wake up asynchronously at the start of the algorithm.

\section{Preliminaries}

\subsection{Matchings}
A \emph{matching} is a subset of the edges of a graph $G$, such that no two of the edges share an endpoint.
We say a matching is \emph{maximum} if it has at least as many edges as any other matching for $G$.
We say a matching is \emph{maximal} if it is not contained in a larger matching for $G$.
Equivalently, a matching is maximal if every edge of $G$ shares at least one endpoint with an edge from the matching.

For $\alpha > 1$, we say a matching is $\alpha$-approximately maximum if its cardinality is at least $1/\alpha$ times the
cardinality of a maximum matching.  It is an immediate consequence of the definitions that any maximal matching is 
$2$-approximately maximum.

\subsection{Maximal vs. Maximum Matchings} \label{sec:max-max}

Perhaps the main reason why maximal matchings are of interest is as an approximate
solution to the related problem of maximum matchings.
Before we begin, we introduce some notations and terminology.

\begin{definition}
We say a matching $M$ is maximal if it is not a subset of any larger matching;
equivalently, if the complementary set of nodes is an independent set.
For a graph $G$, let $\nu(G)$ denote its \emph{matching number}, that is, the maximum 
number of edges in a matching of $G$.  Let $\beta(G)$ denote the minimum number of edges
in a maximal matching of $G$.
Let $\alpha(G)$ denote the \emph{independence number} of $G$, that is, the maximum
size of an independent set (or anti-clique) of $G$.
\end{definition}

% We also need a notation for the size of the worst possible maximal matching.
% \begin{definition}
% Let $\gamma(G)$ denote the number of edges in a minimum-cardinality maximal matching.
% \end{definition}

The following well-known result says that every maximal matching is at least a $\frac12$-approximation
to the size of the maximum matching.
%\begin{definition}

\begin{proposition} \label{prop:max-max}
Let $G$ be any graph.  Then $\beta(G) \ge \frac{\nu(G)}{2}.$
\end{proposition}

The bound in Proposition~\ref{prop:max-max} is tight, as shown for example, by a path of four vertices.
However, for most graphs, it is rather far from tight.  The following bound is due to M. Zito~\cite[Theorem 2]{zito2003small}.

\begin{proposition}
\begin{enumerate}
\item For every graph $G$, we have $\beta(G) \ge \frac{n - \alpha(G)}{2} 
\ge \nu(G) - \frac{\alpha(G)}{2}$.
\item For a random graph $G = G_{n,p}$, where $p = d/n$, 
the inequality 
\[
\beta(G) \ge \frac{n}{2} \left( 1 - \frac{2 \ln(d)}{d} \right)
\]
holds with probability approaching 1 as $d \to \infty$.
\end{enumerate}
\end{proposition}

A number of analogous, related results are proved in~\cite{zito2003small}, generalizing the above to classes of
random bipartite graphs, random regular graphs, and the case where $d$ is a fixed constant, rather than
tending to infinity.  

We mention another kind of random graph which is popular in distributed computing applications, and particularly for radio networks.  These are the so-called \emph{random geometric graphs}, also known as random unit disk graphs.
% We mention one other class of random graphs, which is popular in distributed computing applications, and
% especially radio networks.  
% Another class of graphs of particular interest in the context of radio networks is the so-called random geometric graphs.
For parameters $n,r$, we define the vertex set by choosing $n$ points (vertices) uniformly at random from a square of area $n$.
Two vertices are considered adjacent if their Euclidean distance is less than $r$.  If we neglect boundary effects, this leads
to an average degree of $\pi r^2$.  For such graphs, we can make the following observation.

\begin{proposition}
Let $G = RGG(n,r)$ be a random geometric graph.  Then, 
\[\beta(G) \ge \frac{n}{2} ( 1 - O(1/d)),\] where $d$ is the
expected average degree of $G$.
\end{proposition}

\begin{proof}
Note that the square of area $n$ can be covered by $O(n/r^2)$
disks of radius $r$, hence this is an upper bound on the independence number of
any radius-$r$ disk graph, and in particular a random one.  Thus $\beta(G) \ge \frac{n}{2} - O(\frac{n}{r^2})$.
Since, for $r = \Omega(\sqrt{\log(n)})$, asymptotically almost surely all of the degrees in $G$ are $\Theta(r^2)$, this
shows that $\beta(G) \ge \frac{n}{2} ( 1 - O(1/d) )$, where $d$ is any vertex degree of $G$.
\end{proof}

Taken together, these results show that, in many settings when the graph is not adversarial, maximal matchings may 
be very good approximations to maximum matchings, especially when the average degree is large.

\subsection{Radio Networks and Energy Usage}

We work in the Radio Network model, where we have a communication network on an arbitrary underlying graph $G$. 
Each node in $G$ is a processor equipped with a transmitter and receiver to communicate with other nodes. 
There is an edge between nodes $u$ and $v$ in the graph if $u$ and $v$ are within transmission range of each other. We note that the graph $G$ is not known to the nodes. In fact we will assume that nodes do not know even who their neighbors are in the graph, until they have explicitly heard from them during the running of the algorithm. 
% We will, however assume that the nodes do have a common estimate of home many nodes there are in the graph. In other words, we do not require the nodes to know $n$ explicitly, but we do requite them to agree on a common value $\lceil \log n \rceil$. 

All of the processors begin in the same configuration, although 
we assume they have access to independent sources of random bits.
As a consequence, they can locally generate $O(\log n)$-bit IDs that are unique, with high probability. 
We assume the nodes each know parameters $(n, \Delta)$, where $n$ is an upper bound on the number of nodes in $G$,
and $\Delta$ is an upper bound on the maximum degree of $G$.  It is important for the correctness of 
our algorithm that these values be shared by all nodes, since they act as a kind of synchronization mechanism.
Accuracy of these shared estimates is not needed for correctness, but both running time and energy usage 
depend on these parameters, so if $n$ and $\Delta$ are gross overestimates, it will result in increased costs
for the algorithm.

%\varsha{We should say somewhere that the graph is not known to the individual processors, and that in fact a processors only knows (a subset of) its neighbors after it has heard from them. } 
Time is divided into discrete timesteps. In each timestep a processor can choose to do one of three actions: transmit, listen, or sleep. A message travels from a node $u$ to a neighbor $v$ of $u$ at time $t$ if 
\begin{itemize}
    \item $u$ decides to transmit at time $t$, 
    \item $v$ decides to listen at time $t$ and 
    \item no other neighbor of $v$ decides to transmit at time $t$.
\end{itemize}
Thus when a node $u$ decides to send a message, that message is heard by {\bf all} neighbors of $u$ that happen to be listening, and for whom none of their other neighbors are sending. 

What happens if node $v$ decides to listen and more than one of its neighbors sends a message? There are several different models for this situation. In the most permissive of these, the LOCAL and CONGEST models, $v$ receives all the messages sent by its neighbors. As already specified, we are not working in these models. 
A more restrictive model is the Collision Detection model (CD) where, when a listening node does not receive a message, it can can tell the difference between silence (no neighbors sending) and a collision (more than one neighbor sending).  Another model of interest is the ``No Collision Detection" model (no-CD), which is even more restrictive: here, collisions between two or more messages are indistinguishable from silence.  Prior work~\cite{chang2018energy,chang2020energy} used exponential backoff to deal with collisions in both the CD and no-CD models, making the distinction between these models less important, except in the case of deterministic algorithms.
The maximal matching algorithm in our current paper works in the most challenging (no-CD) model
despite not using backoff.  This can be seen as a corollary of the very local nature of maximal matchings.  

What about message sizes? The LOCAL model allows nodes to send messages of arbitrary size in a single timestep. CONGEST is the same, but with messages restricted to $O(\log n)$ bits.  In our work we follow the message-size constraint 
of the CONGEST model, i.e., each message is $O(\log n)$ bits.

We measure the cost of our algorithms in terms of their energy usage. We assume that a node incurs a cost of 1 energy unit each time that it decides to send \emph{or} listen. When the node is sleeping there is no energy cost. We also assume that local computation is free. The goal of energy aware computation is to design algorithms where the nodes can schedule sleep and communication times so that the energy expenditure is small, ideally $\polylog (n)$, without compromising the time complexity too much, \emph{i.e.,} the running time is still polynomial in $n$.

%\tom{Say something about measuring running time, and the amount of computational work being done per round of the algorithm.}

\section{Notation}

\subsection{The network}

As mentioned earlier, $G = (V,E)$ is the graph defining our radio network.
We denote $n = |V|$, and refer to the nodes as ``processors.''  Although the
processors are identical, and run identical code, we will assume each node
has a unique ID that it knows and uses as its ``name'' in communication.  We
make the standard observation that, if each node were to generate an independently random
string of $C \log n$ bits as its ID, the probability that all $n$ nodes have
distinct IDs is at least $1 - 1/n^{C-2}$, which can be made overwhelmingly likely.

When we present our pseudocode, it will be written from the perspective of a single processor.
However, most of our analysis will be written from the ``global'' perspective of the entire graph.

\subsection{Measuring time}

To begin with, we define two units of time that will be used throughout the paper.
The smaller unit of time is called a \emph{timestep}, and refers to the basic time unit of our radio
network model: in each timestep the nodes that choose to transmit are allowed to send a single message.

The larger unit of time is called a \emph{round}, and consists of three timesteps of the form
$3t-2, 3t-1, 3t$, where $1 \le t \le \tmax$ is the round number.
As shall be seen, rounds have the property that at the end of each round, 
the aggregate state of the network \emph{encodes a matching}.
More precisely, each node has a variable, \partner, and at the end of each round, this variable is either
the ID of one of its neighbor nodes, or has the value {\tt null}; moreover, whenever, at the end of a round, 
\partner$(v) = w \ne $ {\tt null}, we also have \partner$(w) = v$.

\subsection{The Evolving Matching}

For $t \ge 0$, we denote by $M(t)$ the matching encoded by the network at 
the end of round $t$; this is a random
variable whose value is always a pairwise disjoint set of edges of the graph.
As discussed earlier, $M(t)$ is well defined because, at the end of every round,
all vertices have a mutually consistent view of whom they are paired to.

It will be convenient to define some related random variables, all of which are deterministic functions of $M(t)$.
%It will also be useful to have some notation for other random variables that encode the state of the algorithm after each round.
\begin{itemize}
    \item Let $V(t)$ denote the set of unmatched vertices after round $t$.  That is, $V(t) = V \setminus \left(\bigcup_{e\in M(t)} e\right)$.
    \item Let $G(t)$ denote the subgraph of $G$ induced by $V(t)$.  Thus $G(t) = (V(t),E(t))$, where $E(t) = E \cap \binom{V(t)}{2}$. 
    We will refer to this as the \emph{residual graph at the end of round $t$}, or simply \emph{the residual graph.} 
    \item Similarly, for each surviving vertex $v \in V(t)$, we define its \emph{residual neighbor set at the end of round $t$}, $N(v,t) = N(v) \cap V(t)$,
and its \emph{residual degree}, $d(v,t) = |N(v,t)|$.  We denote the closed residual neighborhood of $v$ at the end of round $t$ by $N[v,t]$, defined as $N[v,t] = (N(v) \cup \{v\}) \cap V(t)$.  For matched vertices, $v \notin V(t)$, we adopt the convention $d(v,t) = 0$. 
    \item Finally, we denote the maximum degree in the residual graph
by $\Delta(t) = \displaystyle \max_{v \in V(t)} d(v,t)$, taking this value to be zero if $V(t)$ is empty.
\end{itemize}

We observe that our matching will be non-decreasing over time, that is, for all $t < t'$, 
$M(t) \subseteq M(t')$ with probability one.  It follows that the quantities $|V(t)|$, $\Delta(t)$ and the residual degrees of the individual vertices are all non-increasing in time.

\section{Maximal Matching Algorithm}

%  For all $(v,w) \in E$, let $X_{v,w,t}$ be the indicator random variable for the event that $v$ and $w$ get matched in round $t$ and let $X_v$ be the indicator random variable for the event that $v$ does not matches in a round $t$. Also, $X_{<t}$ is defined as the vector of all $X_{v',w',t'}$ where $t' < t$.

The basic idea of our algorithm is, starting with the empty matching, 
to greedily add disjoint edges until a maximal matching is achieved.
The challenge is to keep each node's energy cost low.
We achieve this by having nodes wake up at random times, and try to
recruit one of their neighbors to pair with them.  If this succeeds without being 
hampered by additional, redundant, neighbors that also happen to wake up, 
then an edge is added to the matching.

To ensure that both endpoints of the edge agree about who they are paired with,
the nodes execute a three-step ``handshake'' protocol, with the property that, 
if it succeeds, both nodes know that the other node has only been in communication 
with them, and was not, for instance, trying to form an edge with another, 
different, endpoint.

To keep the energy costs low, it is essential that nodes wake up with approximately
the correct frequency.  If the rate is too high, 
%(significantly greater than the inverse of the degree), 
too many nodes will wake up at once, causing collisions.
Even if we get around these collisions by some device, having too many nodes wake
up at once seems likely to lead to excessive energy consumption, since at most one
neighbor of a node can get a message through in a single round.

If, on the other hand, the rate is too low,
%(significantly less than the inverse of the degree), then 
too few nodes will wake up at once, again leading to an excessive waste of energy, 
since a node whose neighbors are all asleep cannot form an edge all by itself.

From the perspective of an individual node, whose goal is to connect with exactly one of its neighbors, the ideal would be that, in any given round, it and its neighbors participate with a probability equal to the inverse of its residual degree at the time. There are, however, two problems with setting this to be the participation rate. Firstly, the nodes do not know even their initial degrees, let alone their evolving degrees in the residual graph.  Secondly, even if these degrees were known, nodes of different degrees would desire different participation rates for their neighbors, but their neighbor sets might overlap.

To get around these difficulties, we want to define a \emph{global} participation rate for each round, that acts as a proxy for each node's ideal participation rate. To this end, we define the function
\[
r(t) = \frac{1}{2 + 3\left(1-\frac{t-1}{\tmax}\right)\Delta} 
\]
where $\tmax = C \Delta \log(n)$. 
The constant $C$ will be specified in the
proof of Theorem~\ref{thm:main}.  
% \tom{Pin down the constant, $C$, here.}
This function, $r(t)$ gives a schedule for 
gradually raising the 
participation probability from $r(1) = \frac{1}{2+3\Delta} =\Theta(1/\Delta)$ 
up to $r(\tmax) = \Theta(1)$.

Initially, when the rate is $\Theta(1/\Delta)$ it will be lower than ideal for all but the highest degree vertices. Nevertheless, there is some chance of some pairings being formed. As the algorithm proceeds, the participation probability increases slowly, while a node's residual degree decreases. So for some rounds during the algorithm, the current participation rate (for everyone) will be approximately equal to the inverse of the node's degree, and those are the rounds when the node is most likely to be matched. 

This completes the informal description of our algorithm. 
For a formal specification, Algorithms~\ref{alg:main}, \ref{alg:recruit} and~\ref{alg:accept}
comprise the full pseudocode for our distributed protocol.

\begin{algorithm}
\caption{Main Algorithm:\ A Low-Energy Distributed implementation of Greedy Maximal Matching
in a Radio Network.} \label{alg:main}
\begin{algorithmic}[1]
    \State $t \gets 1$
    \State \partner $ \gets$ {\tt null}
    \While{\partner $=$ {\tt null} and $t \le \tmax$}
        \State Sample $x$ uniformly from $[0,1]$.
        \If{$x \le r(t)/2$}                            
            \State Do RECRUIT\_PROTOCOL this round.
        \ElsIf{$r(t)/2 < x \le r(t)$}
            \State Do ACCEPT\_PROTOCOL this round.
        \Else
            \State Sleep this round.
        \EndIf
        \State $t \gets t + 1$
    \EndWhile
    \State Sleep for the remaining $\tmax - t$ rounds.
\end{algorithmic}
\end{algorithm}

\begin{algorithm}
\caption{RECRUIT$\_$PROTOCOL:\ Try to form an edge as initial sender.} \label{alg:recruit}
\begin{algorithmic}[1]
    \State At timestep 1, Send $my\_ID$ \Comment{``My name is $my\_ID$ and I am available''}
    \State At timestep 2, Listen 
    \If{message received}
        \State Interpret the message as an ordered pair of integers $(x,y)$
        \If{$x = my\_ID$} \Comment{Match found}
            \State \partner $\gets y$
            \State At timestep 3, send $(x,y)$ \Comment{``$x$ and $y$ are paired''}
        \EndIf
    \Else \State Sleep for timestep 3.
    \EndIf
\end{algorithmic}
\end{algorithm}

\begin{algorithm}
\caption{ACCEPT\_PROTOCOL:\ Try to form an edge as initial listener.} \label{alg:accept}
\begin{algorithmic}[1]
    \State At timestep 1, Listen 
    \If{message received}
        \State Interpret the message as an integer $x$
        \State At timestep 2, Send $(x,my\_ID)$ \Comment{``Hello, lets match up, $x$ and $my\_ID$''}
        \State At timestep 3, Listen 
        \If{message received}
            \State Interpret the message as an ordered pair of integers $(x,y)$
            \If{ ($y == my\_ID$) } \Comment{$x$ and $y$ are matched}
                \State \partner $\gets$ $x$
            \EndIf
        \EndIf
    \EndIf
    \State Sleep for any timesteps remaining in the round.
\end{algorithmic}
\end{algorithm}

\section{Maximal Matching Analysis}

In this section we prove the correctness and analyze the running time and energy complexity of Algorithm~\ref{alg:main}.

To begin, we show that the Recruit and Accept protocols run by the individual nodes interact correctly, so that at the end of each round there is no disagreement between nodes about whether or not they are matched and to whom. 

\begin{lemma}
With probability one, at the end of every round $t \ge 0$, the \partner\  variables of 
the $n$ nodes encode a well-defined matching $M(t)$.
\end{lemma}

\begin{proof}
Initially, all the vertices are unmatched, with {\tt null} partners,
so $M(0) = \emptyset$.  Later, we observe that the only circumstances under which the
\partner\  variables have their values reassigned is when 
a vertex $v$ has chosen to participate in that round as recruiter, a neighboring vertex $w$ has chosen
to participate in that round as accepter, and furthermore, both $v$ and $w$ receive a message each time
they Listen during their respective protocols.  Since a message is received if and only if exactly one 
neighbor Sends in that timestep, the messages $v$ receives must come from $w$, and vice-versa.
Therefore $v$ stores the ID of $w$ in its partner variable, and vice-versa.

Furthermore, since $v$ and $w$ would not have participated in round $t$ unless their \partner\  variables
were both {\tt null} beforehand, we know by induction that no other vertices have $v$ or $w$ as their partners.
Since this applies for all vertices and all rounds, the pairing is one-to-one, as desired.
% furthermore, all other neighbors of $v$ have chosen not to participate as
% receivers in that round, and all other neighbors of $w$ have chosen not
% to participate as senders in that round.  \todo{Depending how the algorithm is written, the previous statement may not be true.
% However, the correct direction of implication should still be ok.}
% When this circumstance does occur, \partner$(v)$ changes value from null to $w$,
% and \partner$(w)$ changes value from null to $v$.  Since making any number of such changes 
% preserves the one-to-one property of being a matching, adding a disjoint set of edges to $M(t$), 
% the result follows by induction.
\end{proof}

Now, suppose the algorithm has run for some time, and two neighboring vertices $v$ and $w$ remain unmatched.
The following Lemma gives a fairly tight lower bound on the probability that the edge $\{v,w\}$ will be added to the matching in the next timestep.

\begin{lemma}
 \label{lem:particular-edge}
     Let $t \ge 1$, let $\{v,w\} \in E$, and let $X_{v,w,t}$ be the indicator random variable for the event that $v$ and $w$ get matched to each other in round $t$. Then
     %$\Prob{ X_{v,w,t} = 1} \approx \frac {r^2}2 e^{-2}$.
     \[
     \ExpCond{X_{v,w,t}}{ M(t-1) } \ge 
    \frac{r(t)^2}{2} \left(1 - r(t) \right)^{\Delta(t-1) - 1} \indicator{ v,w \in V(t-1) }
     \]
 \end{lemma}
 
 \begin{proof}
 In order for an edge to form between $v$ and $w$ in round $t$, it is necessary and sufficient for the following four events all to occur:
 \begin{align*}
E_0 &= \{ v, w \in V(t-1) \}, \\
E_1 &= \{ v, w \mbox{ both participate in round $t$ } \}, \\
E_2 &= \{ \mbox{exactly one of $v,w$ participates as recruiter, the other as accepter, in round $t$} \}, \\
E_3 &= \{ \mbox{$E_2$, and $v$ and $w$ receive each others messages without any collision} \}.
 \end{align*}
Note that $E_3 \subset E_2 \subset E_1 \subset E_0$.
For $0 \le i \le 3$, let $X_i = \indicator{E_i}$. 
We now compute expectations, conditioned on the matching at the end of the previous round.
We will prove, below, that
\begin{align}
    \ExpCond{X_1}{M(t-1)} &=  r(t)^2 X_0,  \label{eq:X1} \\
    \ExpCond{X_2}{X_1, M(t-1)} &= \frac{1}{2} X_1, \label{eq:X2} \\
    \ExpCond{X_3}{X_2, X_1, M(t-1)} &\ge \left(1 - r(t) \right)^{\Delta(t-1) - 1} X_2. \label{eq:X3}
\end{align}
 It follows by the law of total expectation that
\[
\ExpCond{X_3}{M(t-1)} \ge \frac{r(t)^2}{2} \left(1 - r(t) \right)^{\Delta(t-1) - 1} X_0,
\]
which is equivalent to the statement of the lemma, noting that  $X_0 = \indicator{ v,w \in V(t-1) }$ and $X_3 = X_{v,w,t}.$
 
To prove the three conditional expectation relations above, first note that equations~\eqref{eq:X1} and \eqref{eq:X2} follow immediately from the definitions of $E_1$ and $E_2$, and the fact that
every vertex in $V(t-1)$ has probability $r(t)/2$ to participate as recruiter in round $t$, and the same probability to participate
as accepter.  

To establish inequality~\eqref{eq:X3}, we note that, conditioned on $E_2$ occurring, for $E_3$ to occur it is 
sufficient\footnote{We note that this is \emph{not} a necessary condition. If $v$ sends a message and two of its neighbors $w$ and $x$ both decide to listen, it could still happen that only $w$ receives the message, because some vertex in $N(x) \setminus N(w)$ sends a message at the same time as $v$,
thereby causing a fortuitous collision at $x$.}
that no other neighbor of $w$ decides to participate in round $t$ in the same role as $v$, and no other neighbor of $v$ decides to
participate in the same role as $w$.  
Thus the conditional probability that $E_3$ occurs is bounded below by the probability that 
\begin{itemize}
    \item no node in $N(v,t-1) \cap N(w,t-1)$ decides to participate at all,
    \item no node in $N(v,t-1) \setminus N[w,t-1]$ decides to participate with the same role as $w$, and
    \item no node in $N(w,t-1) \setminus N[v,t-1]$ decides to participate with the same role as $v$.
\end{itemize}
%Let $A = |N(v,t-1) \cap N(w,t-1)|$ and $B = |N(v,t-1) \setminus N[w,t-1]| + |N(w,t-1) \setminus N[v,t-1]| = |N(v,t-1) \oplus N(w,t-1) \setminus \{v,w\}|$.
Since each node makes its participation decision independently, this probability equals
\[
(1-r(t))^A \left(1-\frac{r(t)}{2}\right)^B  \ge (1 - r(t))^{A + B/2},
\]
where $A = |N(v,t-1) \cap N(w,t-1)|$ and $B = |N(v,t-1) \setminus N[w,t-1]| + |N(w,t-1) \setminus N[v,t-1]|,$
and we have applied the inequality $(1-x/2)^2 \ge 1-x$, which holds for all real $x$.
%$B = |N(v,t-1) \oplus N(w,t-1) \setminus \{v,w\}|$.

\newcommand{\cupdot}{\mathbin{\mathaccent\cdot\cup}}
Next observe that 
\[
N(v,t-1) \setminus \{w\} = \left(N(v,t-1) \cap N(w,t-1)\right) \cupdot \left(N(v,t-1) \setminus N[w,t-1]\right)
\]
and the corresponding equation holds for $N(w,t-1) \setminus \{v\}$, so that 
\begin{align*}
    2A + B  & = |N(v,t-1) \setminus \{w\} | + |N(w,t-1) \setminus \{v\}| \\
    & = d(v, t-1) + d(w, t-1) -2\\
    &\le 2(\Delta(t-1)-1).
\end{align*}
Thus, for a fixed matching $M(t-1)$, the conditional probability of $E_3$ given $E_2$ is at least 
\[
(1 - r(t))^{\Delta(t-1)-1}
\]
This establishes \eqref{eq:X3}, which completes the proof.
% More formally, we have shown the following statement about conditional expectations
% \begin{align*}
%     \ExpCond{X_3}{X_2, X_1, M(t-1)} 
% %    & \ge (1-r(t))^A \left(1-\frac{r(t)}{2}\right)^B X_2 \\
% %    &\ge (1 - r(t))^{A + B/2} X_2 \\
%     & \ge (1 - r(t))^{\Delta(t-1)-1} X_2
% \end{align*}
% which establishes \eqref{eq:X3}, completing the proof.
\end{proof}

Having estimated the probability that a particular matching edge forms at a particular time, we now want to understand the running of the algorithm as a whole. To this end, we make the following definition.

\begin{definition}
Let $1 \le t \le 1+ \tmax$.  We say that the residual graph is good for round $t$ if 
\[\Delta(t-1) r(t) < 1/2.\]
We will also, more concisely, say that $t$ is good, to mean the same thing. 
\end{definition}

Thus, we say ``$t$ is good'' if, prior to round $t$, every vertex $v$ has 
either been matched (and thus $v \notin V(t-1)$)
or enough neighbors of $v$ have been matched to reduce $v$'s residual degree below a target threshold, $\frac{1}{2r(t)}$.
The threshold $\frac{1}{2r(t)}$ was chosen to ensure that any particular vertex listening in round $t$ is unlikely to
miss a message due to a collision.
So, when $t$ is good, any high degree vertex that decides to participate in round $t$ is ``primed to succeed.''
%Intuitively, this means that all the vertices that might potentially participate in round $t$ of the algorithm have residual degrees that are well suited to the current participation rate, so that one may reasonably hope that round $t$ proceeds without too many collisions. 

Our goal will be to prove that, with high probability, $t$ is good for all $1 \le t \le \tmax+1$; 
that is, no degree ever exceeds $\frac{1}{2r(t)}$.  In particular, noting that $r(\tmax + 1) = 1/2$, 
the property of being good for time $\tmax + 1$ means that $\Delta(\tmax) < 1$, which means the final residual graph $G(\tmax)$
is an empty graph; equivalently, $M(\tmax)$ is a maximal matching.

\begin{lemma} \label{lem:good}
Let $A$ be the event that, for all $1 \le t \le \tmax + 1$, the residual graph is good for time $t$.
Then
\[
\Prob{M(\tmax) \mbox{ is maximal }} \ge \Prob{A} \ge 1 - o\left(\frac{1}{n^2}\right).
\]
\end{lemma}

%Theorem~\ref{thm:main} follows immediately from Lemma~\ref{lem:good}

% Note that the residual graph going into round 1 is just the original graph G, and since $r(1)= 1/(2+3\Delta) < 1/2\Delta$, $G$ is good at time 1. In fact there is enough slack here that we will show that because the rate change is gradual, the residual graph remains good for the entire first third of the algorithm (with probability 1). 

% We will then show that in fact, with high probability, the residual graph remains good at all times $1\le t \le 1 + \tmax$.  

In order to prove Lemma~\ref{lem:good} we introduce a random variable that will be used crucially in the remainder of the analysis.
\begin{definition}
For each $1 \le t \le 1+\tmax$ and $v \in V$, let $Z(v,t)$ denote the indicator random variable for the event 
$\left\{ t \mbox{ is good and }
d(v,t) \ge \frac{1}{3r(t)} \right\}$.  By convention, if $v \notin V(t)$, $d(v,t) = 0$, so $Z(v,t) = 0$ also.
\end{definition}
   
The intuition behind this definition is that the event $\{Z(v,t) = 1\}$ means that despite the best
possible conditions for $v$ getting matched: many available unmatched neighbors (since $d(v,t-1) \ge d(v,t) > \frac{1}{3r(t)}$), and a small chance of 
collisions (since $t$ is good), $v$ still failed to get matched in round $t$.  Thus this event represents a lost opportunity for vertex $v$.
% The event $Z(v,t)=1$ means that, in spite of round $t$ being good, vertex $v$ remains an 
% unmatched node with many unmatched neighbors
% at the end of round $t$.  
% In other words most of its neighbors in the residual graph fail to get matched as well.    
Since a vertex cannot get matched in a round unless it participates, which happens with probability only $r(t)$, we must of course
be prepared for many such opportunities to be lost.
However, the following lemma shows that there is a decent chance that any particular such opportunity is \emph{not} lost.

\begin{lemma} \label{lem:Z-condl-upper-bd}
\[
\ExpCond{Z(v,t)}{M(t-1)} \le 1 - \frac{r(t)}{6e}.
\]
\end{lemma}

\begin{proof}
    Recall that, by definition, $Z(v,t) = 1$ if and only if: $v \in V(t)$ and $d(v,t) \ge 1/3r(t)$
    and $\Delta(t-1) < 1/2r(t)$.
    Now, since the degrees at time $t-1$ are determined by $M(t-1)$, and there is nothing to prove
    when the conditional information implies $Z(v,t)$ is identically zero, we may assume 
    $d(v,t-1) \ge 1/3r(t)$ and $\max_w d(w,t-1) \le 1/2r(t)$.
    
    Also, note that $v \notin V(t)$ will occur if and only if $X_{v,w,t} = 1$ for some $w \in N(v,t-1)$.
    Since these events are disjoint, we may sum their probabilities, obtaining
    \begin{align*}
            \ExpCond{Z(v,t)}{M(t-1)} &\le 1 - \sum_{w} \ExpCond{X_{v,w,t}}{M(t-1)} \\
            &\le 1 - \sum_{w\in N(v,t-1)}\frac{r(t)^2}{2} \left(1 - r(t) \right)^{\Delta(t-1)-1} & \mbox{by Lemma~\ref{lem:particular-edge}.}  \\
            &\le 1 - d(v,t-1) \frac{r(t)^2}{2} \left(1 - r(t) \right)^{\Delta(t-1)-1} \\
            &\le 1 - d(v,t) \frac{r(t)^2}{2} \left(1 - r(t) \right)^{1/(2r(t))}
    \end{align*}
where the last inequality follows because $d(v, t-1) \ge d(v,t)$ and $\Delta(t-1)-1 \le 1/(2r(t))$ since $t$ is good. 

    Additionally, since $Z(v,t) = 0$ unless $d(v,t) \ge \frac{1}{3r(t)}$, we have
        \begin{align*}
            \ExpCond{Z(v,t)}{M(t-1)}
            &\le 1 - \frac{1}{3r(t)} \frac{r(t)^2}{2} \left(1 - r(t) \right)^{\frac{1}{2r(t)}} \\
            &\le 1 - \frac{r(t)}{6} \left(1 - r(t) \right)^{\frac{1}{2r(t)}} \\
            &\le 1 - \frac{r(t)}{6 e}.
    \end{align*}
    Here the last inequality follows because for $0 \le x \le 1/2,$  we have $1-x \ge e^{-2 x}$, and $r(t) \le 1/2$ for all $t \le 1+\tmax$. 
\end{proof}

The above lemma shows that in a good round $t$, a particular vertex $v$ has only a bounded chance to ``misbehave''. To prove Lemma~\ref{lem:good} we will show that the first bad round, if any, must be preceded  by a long sequence of good rounds $t'$ on which some vertex $v$ misbehaves (\emph{i.e.,} $Z(v,t')=1$). Since this is unlikely, it must follow that, with high probability, all rounds are good. 

\begin{proof}[Proof of Lemma~\ref{lem:good}]
Assume, for contradiction, that there exists a bad round.  Let $t$ be the first bad round; that is, $t$ is minimal such that 
$\Delta(t) \ge \frac{1}{2r(t)}.$
We note that an easy calculation shows $r(t') < \frac{1}{2\Delta}$ for $t' \le 1 + \tmax/3$, so it must be the case that $t > \tmax/3$.

Consider the set $I = \left\{ t' < t \,\mid \,   3r(t') > 2r(t) \right\}$.  Since $r$ is an increasing function, $I$ is
an interval; let $I = \{t_0, t_0 + 1, \dots, t-1\}$.
Another easy calculation shows that $r(1) < \frac{1}{3\Delta}$, so $t_0 \ge 1$.

% First note that for $1\le t \le 1+\tmax/3$,
% \begin{align*}
%     2r(t) & = \frac{2}{2 + 3\left(1-\frac{t-1}{\tmax}\right)\Delta} \\
%     & \le \frac{2}{2 + 3\left(1-\frac{1}{3}\right)\Delta} \\
%     &= \frac{1}{1+\Delta}\\
%     & \le \frac{1}{\Delta(t-1)}
% \end{align*}
% It follows that in this range, $t$ is good with probability 1. 

% Let $t$ be the first bad round (if any). We know from above that $t$ must be at least $\tmax /3$.
% Consider the set $I$ defined by 
% \[
% I = \left\{ t' < t \,\mid \,   3r(t') > 2r(t) \right\}
% \]
% Since $r$ is a monotone function, $I$ is an interval. \varsha{do we need to explicitly point out that I is not empty? Also is this the right place to point out that the smallest element of I is bigger than 1?}

Since $t$ is by definition the \emph{first} bad round, every $t' \in I$ is good. 
On the other hand, there is a vertex $v$ that is a witness to $t$ being bad, \emph{i.e.,} 
$d(v, t-1) r(t) \ge \frac12 $. Then, for every $t' \in I,$
\[
d(v, t') \ge d(v, t-1) \ge \frac{1}{2r(t)} > \frac{1}{3r(t')}
\]
Combining the two facts above, we conclude that 
\begin{equation}\label{eqn:Z}
   Z(v, t') = 1 \mbox{ for all }t' \in I. 
\end{equation}

Let $\mathcal{E}_{v,t}$ be this event, \emph{i.e.,} the event that $\displaystyle \prod_{t' \in I} Z(v,t') = 1$.  We want to compute the probability of 
$\mathcal{E}_{v,t}$. Recall that $I = \{t_0, \dots t-1\}$. Then

\begin{align*}
    \Prob{\mathcal{E}_{v,t}} &=  \Exp \left( \prod_{t' = t_0}^{t-1} Z(v, t') \right)  \\
    &= \Exp \left( \ExpCond{\prod_{t' = t_0}^{t-1} Z(v, t')}{M(t-2)} \right) & \mbox{by the Law of Total Expectation} \\
    &= \Exp \left( \left( \prod_{t' = t_0}^{t-2} Z(v, t') \right) \; \ExpCond{Z(v, t-1)}{M(t-2)} \right) & \mbox{(*)} \\
    &\le \Exp \left( \prod_{t' = t_0}^{t-2} Z(v, t') \right) \left(1 - \frac{r(t-1)}{6e}\right) & \mbox{by Lemma~\ref{lem:Z-condl-upper-bd}} 
\end{align*}    
Line (*) follows since $M(t-2)$ determines $Z(v,t_0), \dots, Z(v,t-2)$.
% Applying this
% By repeating the steps from using the Law of Total Expectation to applying Lemma~\ref{lem:Z-condl-upper-bd} on the remaining expectation of a product, we see by induction that     
Proceeding inductively, we have

\begin{align*}
    \Prob{\mathcal{E}_{v,t}} &\le \prod_{t' \in I} \left(1 - \frac{r(t')}{6e}\right) \\
    &\le \exp \left( \frac{-1}{6e} \sum_{t' \in I} r(t') \right) & \mbox{since  for all $x$, $1-x \le e^{-x}$}
 \end{align*} 
To get a handle on the expression on the  right hand side, we need a lower bound on the sum of the participation rates. Let $t^{*} \in \mathbb{R}$ be such that $r(t^{*}) = \frac23 r(t)$. Then $t_0 -1 \le t^{*} < t_0$, and we have 

\begin{align*}
    \sum_{t' = t_0}^{t} r(t') & \ge \int_{t_0-1}^{t} r(t') \mathrm{d}t' & \mbox{upper Riemann sum}\\
    &\ge \int_{t^{*}}^{t} r(t') \mathrm{d}t'\\
        &= \int_{t^{*}}^{t} \frac{1}{2 +3(1 -\frac{t'-1}{\tmax})\Delta} \mathrm{d}t'\\
        &= \frac{\tmax}{3\Delta} \int_{1/r(t)}^{1/r(t^*)} \frac{1}{y} \mathrm{d}y  & \mbox{ setting } y = 1/r(t')\\ 
    &= \tfrac{C}{3}\log n \log\left(\frac{r(t)}{r(t^*)}\right) &\mbox{ since $\tmax = C\Delta \log n$} \\
    &= \tfrac{C}{3}\log n \log(3/2)
\end{align*}

But $t \notin I$, so we need to correct the above:
\begin{align*}
\sum_{t' \in I} r(t') & = \left(\sum_{t' = t_0}^{t} r(t')\right) -r(t)\\
& \ge \frac{C}{3} \log n \log(3/2) -\frac12\\
& \ge \frac{C}{8} \log n
\end{align*}

Plugging this back into the probability calculation, 
\begin{align*}
    \Prob{\mathcal{E}_{v,t}} &\le \exp \left( - \frac{1}{6e} \sum_{t' \in I} r(t') \right) \\
    & \le \exp \left(- \frac{ C \log n}{48e} \right) \\
    & \le \frac1{n^{4+\eps}}.
\end{align*}
where the last inequality holds for suitably large values of $C$, \emph{e.g.,} when $C=1000$.

Taking a union bound over the $n \tmax = O(n^2 \log n)$ events $\mathcal{E}_{v,t}$ completes the proof of the lemma.
\end{proof}

Lemma~\ref{lem:good} established that Algorithm~\ref{alg:main} almost surely outputs a maximal matching.
All that remains is to analyze the algorithm's energy cost. 

\begin{proof}[Proof of Theorem~\ref{thm:main}]
The upper bound on energy use comes from a simple analysis of the number of rounds each vertex participates in.
Clearly, the energy use is at most 3 times the number of rounds the vertex participates in, which is at most
the number of heads that would be flipped in $\tmax$ independent coin flips, with probabilities of heads
$r(1), r(2), \dots, r(\tmax)$.
Note that 
\begin{align*}
     \sum_{t=1}^{\tmax -1} r(t)  & \le \int_{1}^{\tmax} r(t) \mathrm{d}t  & \mbox{ lower Riemann sum} \\
        &= \frac{\tmax}{3\Delta} \int_{1/r(\tmax)}^{1/r(1)} \frac{1}{y} \mathrm{d}y  & \mbox{ setting } y = 1/r(t)\\ 
 & = \tfrac{C}{3} \log n \log\left(\frac{r(\tmax)}{r(1)}\right) &\mbox{ since $\tmax = C\Delta \log n$} \\\\
 & = \tfrac{C}{3} \log n \log\left(1+ \tfrac{3}{2}\Delta\right)
\end{align*}
Thus the expected energy use is at most 
\[
3 \sum_{t=1}^{\tmax} r(t) = r(\tmax) + 3 \sum_{t=1}^{\tmax-1} r(t) \le C \log n \log\left(1+ \tfrac{3}{2}\Delta\right) +\tfrac{1}{2} = O\big((\log n) (\log \Delta)\big)
\]
% \approx 9 C \log n \log(4 C \log n + \tmax / (4 C \log n))) $
% which is about $9 C \log n \log((4+n)/4) = O(\log^2 n)$.
Chernoff's bound, together with with a union bound over the $n$ vertices,
implies the high-probability upper bound on expected energy cost.
\end{proof}

\section{Neighbor Assignment Functions}\label{sec:buddies}

Motivated by the problem of assigning nodes to backup data from their neighbors in a sensor network, 
we introduce the following definition.  As we shall see later, it is extremely closely connected to 
the established concept of matching covering number.

\begin{definition}
Given graph $G = (V,E)$, a \emph{neighbor assignment function (NAF)} is a function $f : V \to V$ such that for all $v \in V$, 
$\{v, f(v)\} \in E$.  Equivalently, we may think of this as an oriented subgraph of $G$, in which each vertex has out-degree 1. 
The \emph{load} of the assignment is the maximum in-degree of this digraph.  Equivalently, load is $\max_{v \in V} |f^{-1}(v)|$.
The \emph{minimum NAF load} of $G$ is the minimum load among all NAFs for $G$.  
\end{definition}

\noindent {\bf Note:} In the case when $G$ is bipartite, NAFs are also known as ``semi-matchings.''  (See, for example, 
\cite{harvey2006semi,czygrinow2016distributed}.)
However, since we are particularly concerned with the non-bipartite case, we preferred to introduce a different term.

In the context of backing up data, we think of the assigned node $f(v)$ as the node who will store a backup copy of $v$'s data.
Our goal for this section is to find a NAF whose load is small.  In the energy-aware radio network
setting, we also want to ensure that the per-node energy use is small.

% The undirected edges correspond to vertices that back each other up. The directed edges point from vertices needing to be backed up to vertices that are backing them up. Thus the in-degree of a vertex corresponds to its backup load. We want the maximum such load to be no more than $L = o(n)$, preferably constant.  
% Note that such a structure may not exist for a graph. For instance when $G$ is a star, the maximum load is $n-1$.

Our next result establishes a close connection between the load of the best NAF for a graph and the minimum
number of matchings needed to cover all of its vertices.

\begin{definition}
The \emph{matching cover number} of a graph $G$, denoted $\mathrm{mc}(G)$, is the minimum integer $k$ such that there exists a set of $k$
matchings of $G$, 
whose union contains every vertex of $G$.
\end{definition}

\begin{theorem} \label{thm:NAF-mc}
    For every graph $G$, the minimum NAF-load of $G$ equals the matching cover number of $G$, unless the
    NAF-load of $G$ equals 1.  If the NAF-load of $G$ equals $1$, the matching cover number of $G$ can be 
    $1$ or $2$.
\end{theorem}

\begin{proof}
  Suppose $V=V(G)$ is covered by the union of matchings $M_1, \dots, M_L$.  Then assigning each vertex
  $v$ to its partner in the first matching that contains $v$ is an NAF with maximum load at most $L$.
  This establishes that the NAF-load is always at most the matching covering number.
  
  Before we begin the proof for the reverse implication, we make the following general observation about
  digraphs with out-degree 1.  By considering the unique walk obtained by starting at any vertex $v$, and
  repeatedly following the edge $\{v,f(v)\}$, we can see that each weakly connected component consists of
  one oriented cycle (of length $\ge 2$), together with one or more ``tributary'' trees, each rooted at a 
  node of this cycle, and oriented towards that root.  See Figure~\ref{fig:tributaries}.
  
  \begin{figure}
      \centering
      \includegraphics[scale=0.3]{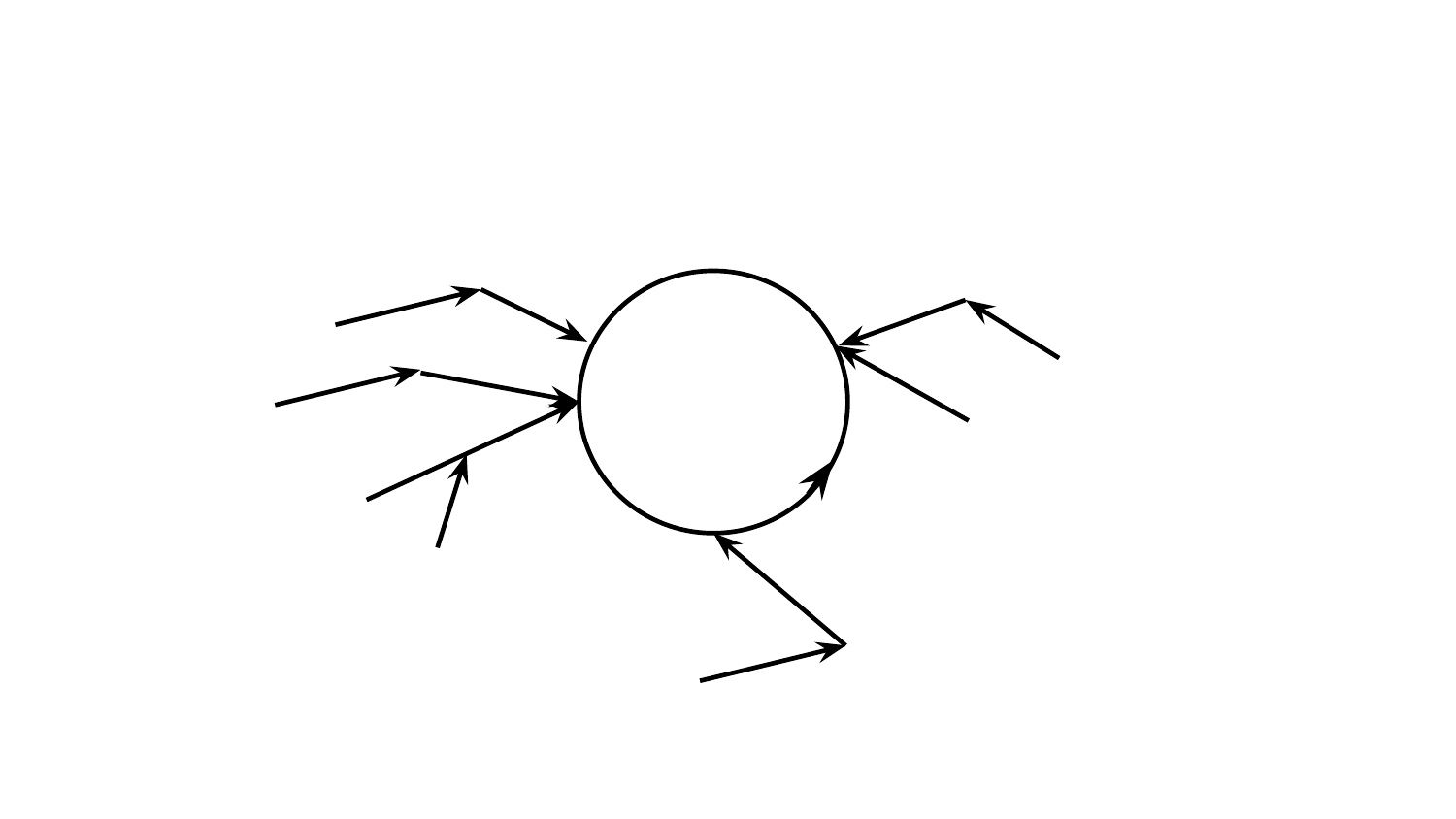}
      \caption{Sketch of the digraph of a NAF with only one component.  Each component can be seen as a 
      directed cycle, along with zero or more ``tributary trees.''}
      \label{fig:tributaries}
  \end{figure}
  
  Now, if $f$ has any leaf, that is, a node $v$ whose load is zero, we can obtain a new NAF by 
  reassigning $f(v)$ to point back to $v$.  This increases the load at $v$ to 1, decreases the load by 1 at $f(f(v))$, 
  and does not change any other vertex loads.  Repeated application of this rule to all leaves in turn, eventually leads
  to a NAF whose components are all either (a) directed cycles, which do not have any leaves, or (b) stars with one bi-directed edge.
  See Figure~\ref{fig:NAF-reduction}.
  In case (b), the component consists of one node, $r$, of in-degree $i \le L$, $i$ nodes, $x_1, \dots, x_i$, each with an edge
  directed to $r$, and one edge from $r$ to $x_1$. 

\begin{figure}
    \centering
    \includegraphics[scale=0.3]{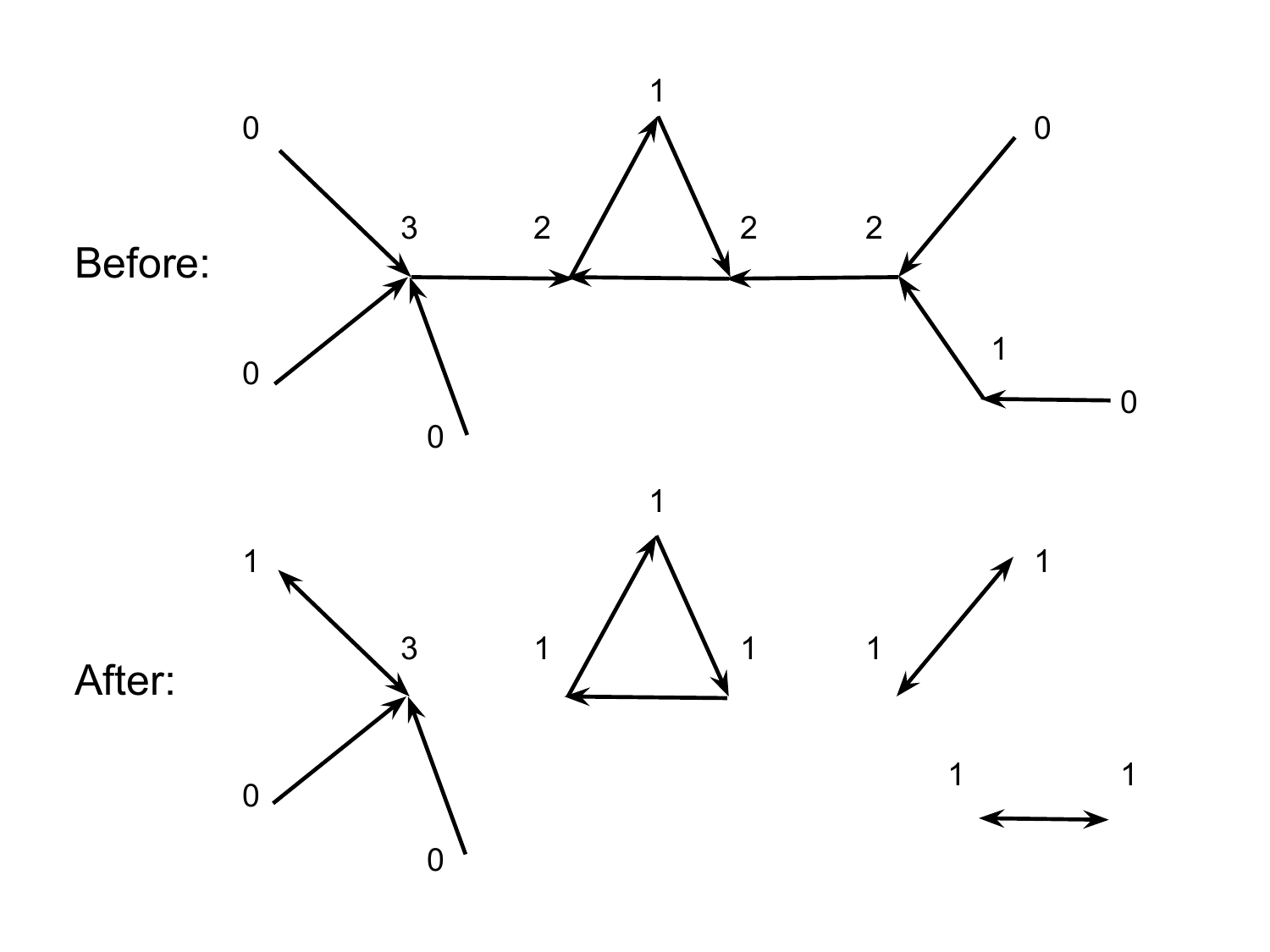}
    \caption{Example of the conversion of a given NAF by decreasing the number of load-zero nodes.
    The number by each node indicates its load.  Note that, after the conversion, the maximum load did not increase,
    and the connected components of the NAF are now all directed cycles and/or stars with one bi-directed edge. 
    (The components of size 2 are both.)}
    \label{fig:NAF-reduction}
\end{figure}

  It is easy to see that, for a directed cycle, whose edges are $e_1, \dots, e_{\ell}$, 
  a single matching consisting of the even edges, $e_2, e_4, \dots,$, will cover all the vertices if $\ell$ is even,
  and all but one vertex if $\ell$ is odd.  Therefore, one matching covers the component if $\ell$ is even, and
  two if $\ell$ is odd.
  
  For the star with bi-directed edge, the maximum load equals the degree, $d$, of the center vertex.  And a matching cover
  consists of the $d$ single edges that make up the star.
  
  In this way, we can build up our matching cover component by component, noting that if every component has a matching cover
  of size at most $k$, then so does the entire graph.  Since the only case when our matching cover was bigger than the 
  maximum load for the component was when $L=1$, the proof is complete.
\end{proof}

Wang, Song, and Yuan~\cite{wang2014matching} have given an $O(n^3)$-time centralized algorithm for finding the minimum number of matchings needed
to cover a graph.  In light of Theorem~\ref{thm:NAF-mc}, their result implies an $O(n^3)$ time algorithm for finding the minimum-load NAF for any graph.

In the distributed and low-energy setting, it is unlikely that we can achieve such an ambitious goal.  For instance, a node cannot determine
its exact degree without sending and/or receiving at least that many messages successfully, which may require linear energy.  Instead, we aim for the less ambitious goal
of finding a NAF whose maximum load is well within our energy budget.  Our next result shows that this is possible, assuming one exists.

\begin{algorithm}
\caption{Low-Energy Distributed algorithm to compute a NAF
in a Radio Network.} \label{alg:NAF}
\begin{algorithmic}[1]
    \State{Run our Maximal Matching algorithm on $G$.}
    \State{For each edge $\{u,v\}$ in the matching, mark $u,v$ as assigned, and assign them to each other.}
    \For{ $i \gets 1$ to $k$} 
         \State{Run the maximal matching algorithm on $G$, modified so that only unassigned nodes \phantom . \hspace{2em} are allowed to recruit,
         and only assigned nodes are allowed to accept.}
         \State{For each edge $\{u,v\}$ in the matching, mark $u,v$ as assigned, and (re-)assign them to \phantom . \hspace{2em} each other.}
    \EndFor
\end{algorithmic}
\end{algorithm}

First however, we need another definition.

\begin{definition}
For a graph $G$, a partial NAF is a function $f: S \to V$, where $S \subseteq V$.
As before, we define the \emph{maximum load} of $f$ as $\max_{v \in V} |f^{-1}(v)|$.
We say the \emph{coverage} of $f$ is $|S|/|V|$.
\end{definition}

Our motivation for introducing partial NAF's stems from the following possibility.
A particular graph $G$ may not have any NAF's whose maximum load is less than its maximum degree, $\Delta(G)$.
Despite this, it is possible that, say, 90\% of its vertices would be satisfied by a partial NAF whose maximum load 
is $1$.  In this case, we might prefer the partial NAF to the best complete one, in spite of the unassigned vertices.
Our next result shows that running Algorithm~\ref{alg:NAF} should produce a result that is, in some sense, 
competitive with every partial NAF for $G$.

% Our next result establishes that, we can achieve the more modest goal of an NAF whose load is $O(\log n)$ times the optimal load, 
% using only polylog times load energy in
% the distributed radio network setting.
\newcommand{\eee}{\mathrm{e}}
\begin{theorem} \label{thm:NAF-alg}
Let $\eps \ge 0$, and let $G$ be a graph for which there exists a partial NAF with coverage $1-\varepsilon$
and maximum load $L$.
Then Algorithm~\ref{alg:NAF}, run with parameter $k$, will, with probability $1 - O\left(\frac{k}{n^2}\right)$, output a 
partial NAF with coverage $(1-\varepsilon)(1-e^{-k/(2L+2)})$ and maximum load at most $k$.
Its per-vertex expected energy usage is $O(k \log^2 n)$.  In particular, if $k \ge (2L+2)\log(n)$,
the output NAF will also have coverage $1-\eps$.
\end{theorem}

\begin{proof}
Let $f$ be a partial NAF with coverage $1 - \varepsilon$ and maximum load $L$.
First we convert $f$ into a complete NAF on a subgraph of $G$.  Let $S$ be the 
domain of $f$, and let $R$ be the range of $f$.  We extend $f$ to the domain $S \cup R$
by, for every vertex $v \in R \setminus S$, arbitrarily choosing a vertex $w \in f^{-1}(v)$,
and defining $f(v) = w$.  Since a different $w$ is necessarily chosen for each $v \in R \setminus S$,
this increases the load of $f$ by at most $1$.

Now that $f$ is a NAF for the subgraph induced by $S \cup R$, we apply Theorem~\ref{thm:NAF-mc}
to deduce the existence of a matching cover of size $L+1$ that includes every vertex of $S \cup R$.
This implies that the maximum matching covers at least $(1-\varepsilon)n /(L+1)$ vertices.  
Hence every maximal matching covers at least $(1-\varepsilon)n/(2L+2)$ vertices.
So the first call to the maximal matching algorithm will assign neighbors to at least this many vertices.

In subsequent rounds, the modification to the maximal matching algorithm has the effect of making
it run on the bipartite graph where the bipartition is into the assigned and unassigned vertices.
By the pigeonhole principle, at least one matching, $M$, from the $L+1$ in the matching cover
must cover at least a $1/(L+1)$ fraction of the unassigned vertices in $S \cup R$.  
Since the first matching was maximal, no edges in $G$ have both endpoints unassigned; therefore, 
$M$ is a matching within the bipartite graph being fed into our maximal matchings algorithm.  
Therefore, the maximal matching that is found must cover at least a $1/2(L+1)$ fraction of the unassigned
vertices.  It follows that after $k$ iterations, at most 
\[
\left( 1 - \frac{1}{2L+2} \right)^k (1 - \varepsilon) n \le e^{-k/(2L+2)} (1 - \varepsilon) n
\]
nodes from $S \cup R$ will remain unassigned.

Since each run of the maximal matching algorithm succeeds 
with probability $1 - O(1/n^2)$, a union bound over the $k$ outer loop iterations 
establishes the high-probability bound.  
\end{proof}

We point out that, at the end of each loop iteration of Algorithm~\ref{alg:NAF}, any assigned vertices
that were not matched with an unassigned node in that iteration must have no unassigned neighbors, and
can therefore go to sleep for the rest of the algorithm.  If desired, Algorithm~\ref{alg:NAF} can even
be run with parameter $k$ set to $\infty$, since the algorithm will now terminate once a NAF is found.

\section*{Acknowledgments}

The authors would like to thank the anonymous referees of the conference version of this paper
for helpful comments and suggestions.

\bibliography{match}

\end{document}